\def\itcs{0}
\def\notes{0}

\ifnum\itcs=1 \def\notes{0} \fi

\documentclass[11pt]{article}

\usepackage[section]{algorithm}
\usepackage[noend]{algpseudocode}
\algrenewcommand\algorithmicrequire{\textbf{Input:}}
\algrenewcommand\algorithmicensure{\textbf{Output:}}
\algrenewcommand\algorithmicwhile{\textbf{While}}
\algrenewcommand\algorithmicfor{\textbf{For}}
\algrenewcommand\algorithmicreturn{\textbf{Return}}
\algrenewcommand\algorithmicif{\textbf{If}}
\algnewcommand\algorithmicconst{\textbf{Constants:}}
\algnewcommand\Const{\item[\algorithmicconst]}

\newcommand{\Sasoglu}{\c{S}a\c{s}o\u{g}lu}
\newcommand{\Arikan}{Arikan}
\newcommand{\poly}{\mathop{\mathrm{poly}}}
\newcommand{\Ch}{\mathcal{C}}
\newcommand{\bH}{\bar{H}}
\newcommand{\cY}{\mathcal{Y}}
\newcommand{\td}{\widetilde}
\newcommand{\cD}{\mathcal{D}}
\newcommand{\cS}{\mathcal{S}}
\newcommand{\bU}{\bar{U}}
\newcommand{\Oh}{\mathcal{O}}

\newcommand{\Ziid}{\bar{Z}}

\newcommand{\HMM}{\mathcal{H}}

\usepackage{style}
\usepackage[margin=1in]{geometry}

\title{Algorithmic Polarization for Hidden Markov Models}
\ifnum\itcs=0
\author{%
Venkatesan Guruswami\thanks{Computer Science Department, Carnegie Mellon University, Pittsburgh, PA 15213. {\tt venkatg@cs.cmu.edu}. Most of this work was done when the author was visiting the Center for Mathematical Sciences and Applications, Harvard University, Cambridge, MA. Research supported in part by NSF grants CCF-1422045 and CCF-1814603.}
\and
Preetum Nakkiran\thanks{Harvard John A. Paulson School
of Engineering and Applied Sciences, Harvard University, 33 Oxford Street,
Cambridge, MA 02138, USA. Email: {\tt preetum@cs.harvard.edu}. Work supported in
part by the NSF Graduate Research Fellowship Grant No. DGE1144152,
and Madhu Sudan's Simons Investigator Award and NSF Award CCF 1715187.}
\and
Madhu Sudan\thanks{Harvard John A. Paulson School of Engineering and Applied Sciences, 33 Oxford Street, Cambridge, MA 02138, USA. {\tt madhu@cs.harvard.edu.} Work supported in part by a Simons Investigator Award and NSF Award CCF 1715187.}
}
\fi
\date{\today}
\begin{document}
\maketitle
\thispagestyle{empty}
\mnote{If this note is appearing then so are author names and other notes. Remove both before final submission by setting itcs flag to 1.}

\begin{abstract}

Using a mild variant of polar codes we design linear compression schemes compressing Hidden Markov sources (where the source is a Markov chain, but whose state is not necessarily observable from its output), and to decode from Hidden Markov channels (where the channel has a state and the error introduced depends on the state). We give the first polynomial time algorithms that manage to compress and decompress (or encode and decode) at input lengths that are polynomial \emph{both} in the gap to capacity and the mixing time of the Markov chain. Prior work achieved capacity only asymptotically in the limit of large lengths, and polynomial bounds were not available with respect to either the gap to capacity or mixing time.
Our results operate in the setting where the source (or the channel) is {\em known}. If the source is {\em unknown} then compression at such short lengths would lead to effective algorithms for learning parity with noise --- thus our results are the first to suggest a separation between the complexity of the problem when the source is known versus when it is unknown.

\end{abstract}

\newpage

\ifnum\notes=1
TODOs:
\begin{itemize}
\item Many comments and remarks (non-stationary, limiting entropy vs $H(Z^n)$, etc.)
\end{itemize}

\fi

\section{Introduction}
\label{sec:intro}

We study the problem of designing coding schemes, specifically encoding and decoding algorithms, that overcome errors caused by stochastic, but not memoryless, channels. Specifically we consider the class of ``(hidden) Markov channels'' that are stateful, with the states evolving according to some Markov process, and where the distribution of error depends on the state.\footnote{We use the term {\em hidden} to emphasize the fact that the state itself is not directly observable from the actions of the channel, though in the interest of succinctness we will omit this term for most of the rest of this section.} Such Markovian models capture many natural settings of error, such as bursty error models. (See for example, Figure~\ref{fig:one}.) Yet they are often less understood than their memoryless counterparts (or even ``explicit Markov models'' where the state is completely determined by the actions of the channel). For instance (though this is not relevant to our work) even the capacity of such channels is not known to have a closed form expression in terms of channel parameters. (In particular the exact capacity of the channel in Figure~\ref{fig:one} is not known as a function of $\delta$, $p$ and $q$!)

\begin{figure}[h]
		\centering\includegraphics[scale=.5]{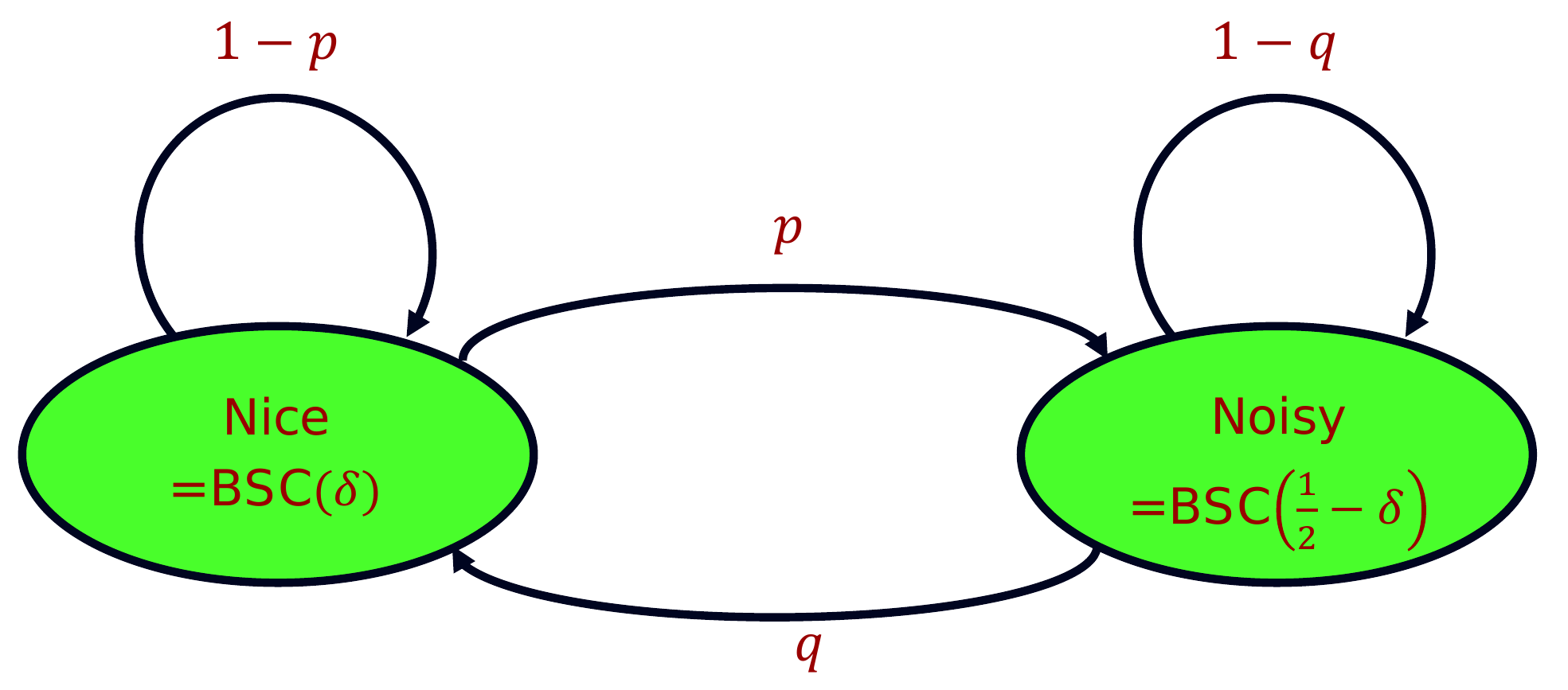}
	\caption{A Markovian Channel: The Nice state flips bits with probability $\delta$ whereas the Noisy state flips with probability $1/2 - \delta$. The stationary probability of the Nice state is $q/p$ times that of the Noisy state.}\label{fig:one}
\end{figure}

In this work we aim to design coding schemes that achieve rates arbitrarily close to capacity. Specifically given a channel of capacity $C$ and gap parameter $\epsilon > 0$, we would like to design codes  that achieve a rate of at least $C - \epsilon$, that admit polynomial time algorithms even at small block lengths $n \geq \poly(1/\epsilon)$. Even for the memoryless case such coding schemes were not known till recently. In 2008, \Arikan~\cite{arikan-polar} invented a completely novel approach to constructing codes based on ``channel polarization" for communication on binary-input memoryless channels, and proved that they enable achieving capacity in the limit of large code lengths with near-linear complexity encoding and decoding. 
 In 2013, independent works by Guruswami and Xia~\cite{GX15} and Hassani et al.~\cite{HAU14} gave a finite-length analysis of \Arikan's polar codes, proving that they approach capacity fast, at block lengths
bounded by $\poly(1/\epsilon)$ where $\epsilon > 0$ is the difference between the channel capacity and code rate.

The success of polar codes on the memoryless channels might lead to the hope that maybe these codes, or some variants, might lead to similar coding schemes for channels with memory. But such a hope is not easily justified: the analysis of polar codes relies heavily on the fact that errors introduced by the channel are independent and this is exactly what is not true for channels with memory. Despite this seemingly insurmountable barrier, \Sasoglu~\cite{Sasoglu} and later \Sasoglu\ and Tal~\cite{SasogluTal}
 showed, quite surprisingly, that the analysis of polar codes can be carried out even with Markovian channels (and potentially even broader classes of channels). Specifically they show that these codes converge to capacity and even the probability of decoding error, under maximum likelihood decoding, drops exponentially fast in the block length (specifically as $2^{-n^{\Omega(1)}}$ on codes of length $n$; see also \cite{ShuvalTal}, where exponentially fast polarization was also shown at the high entropy end). An extension of \Arikan's successive cancellation decoder from the memoryless case was also given by \cite{WHYLH}, building on an earlier version~\cite{WLH} specific to intersymbol interference channels, leading to efficient decoding algorithms.
 

However, none of the works above give small bounds on the block length of the codes as a function of the gap to capacity, and more centrally to this work, on the mixing time of the Markov chain.
The latter issue gains importance when we turn to the issue of ``compressing Markov sources'' which turns out to be an intimately related task to that of error-correction for Markov channels as we elaborate below and which is also the central task we turn to in this paper. We start by describing Markov source and the (linear) compression problem.

A (hidden) Markov source over alphabet $\Sigma$ is given by a Markov chain on some finite state space where each state $s$ has an associated distribution $D_s$ over $\Sigma$. The source produces information by performing a walk on the chain and at each time step $t$, outputting a letter of $\Sigma$ drawn according to the distribution associated with the state at time $t$ (independent of all previous choices, and previous states).\footnote{The phrase ``hidden'' emphasizes the fact that the output produced by the source does not necessarily reveal the sequence of states visited.}
In the special case of additive Markovian channels where the output of the channel is the sum of the transmitted word with an error vector produced by a Markov source, a well-known correspondence shows that error-correction for the additive Markov channel reduces to the task of designing a compression and decompression algorithm for Markovian sources, with the compression being {\em linear}. Indeed in this paper we only focus on this task: our goal turns into that of compressing $n$ bits generated by the source to its entropy upto an additive factor of $\epsilon n$, while $n$ is only polynomially large in $1/\epsilon$.

A central issue in the task of compressing a source is whether the source is {\em known} to the compression algorithm or not. While ostensibly the problem should be easier in the ``known'' setting than in the ``unknown'' one, we are not aware of any formal results suggesting a difference in complexity. It turns out that compression in the setting where the source is unknown is at least as hard as ``learning parity with noise'' (we argue this in Appendix~\ref{app:LPN}), {\em if} the compression works at lengths polynomial in the mixing time and gap to capacity. This suggests that the unknown source setting is hard (under some current beliefs). No corresponding hardness was known for the task of compressing sources when they are known, but no easiness result seems to have been known either (and certainly no linear compression algorithm was known). This leads to the main question addressed (positively) in this work.

\paragraph{Our Results.}
Our main result is a construction of codes for {\em additive Markov channels} that gets $\epsilon$ close to capacity at block lengths polynomial in $1/\epsilon$ and the mixing time of the Markov chain, with polynomial (in fact near-linear) encoding and decoding time. Informally additive channels are those that map inputs from some alphabet $\Sigma$ to outputs over $\Sigma$ with an abelian group defined on $\Sigma$ and the channel generates an error sequence independent of the input sequence, and the output of the channel is just the coordinatewise sum of the input sequence with the error sequence. (In our case the alphabet $\Sigma$ is a finite field of prime cardinality.) The exact class of channels is described in Definition~\ref{def:amc}, and Theorem~\ref{thm:main-channel} states our result formally. We stress that we work with additive channels only for conceptual simplicity and that our results should extend to more general symmetric channels though we don't do so here. Prior to this work no non-trivial Markov channel was known to achieve efficient encoding and decoding at block lengths polynomial in either parameter (gap to capacity or mixing time).


Our construction and analyses turn out to be relatively simple given
the works of \Sasoglu\ and Tal~\cite{Sasoglu,SasogluTal} and the work
of Blasiok et al.~\cite{BGNRS}. The former provides insights on how to
work with channels with memory, whereas the latter provides tools
needed to get short block length and cleaner abstractions of the
efficient decoding algorithm that enable us to apply it in our
setting. Our codes are a slight variant of polar codes, where we apply
the polar transforms independently to blocks of inputs. This enables
us to apply the analysis of \cite{BGNRS} in an essentially black box manner,
benefiting both from its polynomially fast convergence guarantee to
capacity as well as its generality covering all polarizing matrices over any prime alphabet (and not just the basic Boolean $2 \times 2$ transform covered in \cite{SasogluTal}). 

We give a more detailed summary of how our codes are obtained and how
we analyze them in Section~\ref{sec:overview} after stating our
results and main theorem formally.

\section{Definitions and Main Results}

\subsection{Notation and Definitions}

We will use $\F_q$ to denote the finite field with $q$ elements. Throughout the paper, we will deal only with the case when $q$ is a prime. (This restriction in turn comes from the work of \cite{BGNRS} whose results we use here.)

We use several notations to index matrices.
For a matrix $M \in \F_q^{m \x n}$, the entry in the $i$th row, $j$th column is
denoted $M_{i, j}$ or $M_{(i, j)}$.
Columns are denoted by superscripts, i.e., $M^j \in \F_q^m$ denotes the $i$th
column of $M$. Note that $M^j_i = M_{(i, j)}$.
We also use the indices as sets in the natural way. For example
$M^{\leq j} \in \F_q^{m \times j}$ denotes the first $j$ columns of $M$.
$M^{\leq j}_{\leq i}$ denotes the submatrix of elements in the first $j$ columns
and first $i$ rows.
$M_{\prec (i, j)}$ denotes the set of elements of $M$ indexed by
lexicographically smaller indices than $(i, j)$.
Multiplication of a matrix $M \in \F_q^{m \x n}$ with a vector $v \in \F_q^n$ is
denoted $Mv$.

For a finite set $S$, let $\Delta(S)$
denote the set of probability distributions over $S$.
For a random variable $X$ and event $E$, we write $X | E$ to denote the
conditional distribution of $X$, conditioned on $E$.
For example, we may write $X|\{X_1 = 0\}$.

The \emph{total-variation distance} between two distributions $p, q \in
\Delta(U)$ is
$$||p- q||_1 := \sum_i |p(i) - q(i)|$$

We consider compression schemes, as a map $\F_q^n \to \F_q^m$.
The \emph{rate} of a compression scheme $\F_q^n \to \F_q^m$ is the ratio $m/n$.

For a random variable $X \in [q]$,
the \emph{(non-normalized) entropy} is denoted $H(X)$, and is
$$H(X) := -\sum_{i} \Pr[X = i] \log(\Pr[X = i])$$
and the \emph{normalized entropy} is denoted $\bH(X)$, and is
$$\bH(X) := \frac{1}{\log(q)}H(X)$$

\begin{definition}
A \emph{Markov chain} $\mathcal{M} = (\ell,\Pi,\pi_0)$ is given by an $\ell$ representing the state space $[\ell]$, a transition
matrix $\Pi \in \R^{\ell \x \ell}$, and a distribution on initial state
$\pi_0 \in \Delta([\ell])$. The rows of $\Pi$, denoted $\Pi_1,\ldots,\Pi_\ell$ are thus elements of $\Delta([\ell])$.
A Markov chain generates a random sequence of states $X_0, X_1, X_2, \dots$
determined by letting $X_0 \sim \pi_0$, and $X_t \sim \Pi_{X_{t-1}}$ for
$t > 0$ given $X_0,\ldots,X_{t-1}$.
The stationary distribution $\pi \in \Delta([\ell])$ is the distribution such
that if $X_0 \sim \pi$, then all $X_t$'s are marginally identically distributed as $\pi$.
\end{definition}

We consider only Markov chains which are irreducible and aperiodic, and hence have a stationary distribution to which they converge in the limit. The rate of convergence is measured by the mixing time, defined below.

\begin{definition}
The \emph{mixing time} of a Markov chain is the constant
$\tau > 0$ such that for every initial state $s_0$ of the Markov chain,
the distribution of state $s_\ell$ is $\exp(-\ell / \tau)$-close in total
variation distance to the
stationary distribution $\pi$.
\end{definition}

\begin{definition}
A \emph{(stationary, hidden) Markov source} $\HMM = (\Sigma,\mathcal{M},\{\cS_1,\ldots,\cS_\ell\})$ is specified by an alphabet $\Sigma$, a Markov chain $\mathcal{M}$ on $\ell$ states and distributions $\{\cS_i \in \Delta(\Sigma)\}_{i \in [\ell]}$. The output of the source is a sequence $Z_1,Z_2,\ldots,$ of random variables obtained by first sampling a sequence $X_0,X_1,X_2,\ldots$ according to $\mathcal{M}$ and then sampling $Z_i \sim \cS_{X_i}$ independently for each $i$.
We let $\HMM_t$ the distribution of output sequences of length $t$, and $\HMM_t^{\otimes s}$ denote the distribution of $s$ i.i.d. samples from $\HMM_t$.
\end{definition}

Similarly, we define an \emph{additive Markov channel} as a channel which adds
noise from a Markov source.

\begin{definition}\label{def:amc}
An \emph{additive Markov channel $\Ch_{\HMM}$}, specified by a Markov source $\HMM$ over alphabet $\F_q$,  is a randomized map
$\Ch_{\HMM}: \F_q^* \to \F_q^*$ obtained as follows:
On channel input $X_1,\ldots,X_n$, the channel outputs $Y_1,\ldots,Y_n$ where $Y_i = X_i + Z_i$ where $Z = (Z_1,\ldots,Z_n) \sim \HMM_n$.
\end{definition}

\begin{definition}
A \emph{linear code} is a linear map $C: \F_q^k \to \F_q^n$.
The \emph{rate} of a code is the ratio $k/n$.
\end{definition}


\begin{definition}
\label{def:constr}
For all sets $A, B$, a \emph{constructive source over $(A | B)$ samplable in
time $T$} is a
distribution $\cD \in \Delta(A \times B)$ such that $(a, b) \sim \cD$ can be sampled
efficiently in time at most $T$,
and for every fixed $b \in B$, the conditional distribution $A | \{B = b\}$ can be
sampled efficiently in time at most $T$.
\end{definition}


\begin{proposition}
    Every Markov source with state space $[\ell]$ is a constructive source
    samplable in time $\Oh(n\ell^2)$.
    That is, for every $n$, let $Y_1, \dots Y_n$ be the random variables generated by the
    Markov source.
    Then, the sequence $Y_1, \dots Y_n$ can be sampled in time at most
    $\Oh(n \ell^2)$, and moreover for every setting of $Y_{< n} = y_{< n}$, the
    distribution $(Y_n | Y_{< n} = y_{<n})$ can be sampled in time $\Oh(n
    \ell^2)$.
\end{proposition}

\begin{proof}
Sampling $Y_1, \dots, Y_n$ can clearly be done by simulating the Markov chain,
and sampling from the conditional distribution $(Y_n | Y_{< n} = y_{<n})$
is possible using the standard \emph{Forward Algorithm} for inference in Hidden Markov Models, which we describe for completeness in Appendix~\ref{app:forward-algo}.
\end{proof}

Finally, we will use the following notion of \emph{mixing matrices} from
~\cite{KSU10,BGNRS}, characterizing which matrices lead to good polar codes. In the study of polarization it is well-known that lower-triangular matrices do not polarize at all, and the polarization characteristics of matrices are invariant under column permutations. Mixing matrices are defined to be those that avoid the above cases.

\begin{definition}
	\label{def:mixing}
For prime $q$ and $M \in \F_q^{k \x k}$, $M$ is said to be a
\emph{mixing matrix} if $M$ is invertible and for every permutation of the columns
of $M$, the resulting matrix is not lower-triangular.
\end{definition}

\subsection{Main Theorems}

We are now ready to state the main results of this work formally.
We begin with the statement for compressing the output of a hidden Markov model.

\begin{theorem}
\label{thm:main}
For every prime $q$ and mixing matrix $M \in \F_q^{k \x k}$ there exists a
preprocessing algorithm ({\sc Polar-Preprocess}, Algorithm~\ref{algo:subset}),
a compression algorithm ({\sc Polar-Compress}, Algorithm~\ref{algo:compressor}),
a decompression algorithm ({\sc Polar-Decompress}, Algorithm~\ref{algo:fast-decompressor})
and a polynomial $p(\cdot)$ such that for every $\eps > 0$, the following properties hold:
\begin{enumerate}
\item {\sc Polar-Preprocess} is a randomized algorithm
that takes as input a
Markov source $\HMM$ with $\ell$ states, and $t \in \mathbb{N}$, and runs in time $\poly(n, \ell, 1/\eps,q)$ where $n=k^{2t}$
 and outputs
auxiliary information for the compressor and decompressor (for $\HMM_n$).

\item {\sc Polar-Compress} takes as input a sequence $Z \in \F_q^n$
as well as the auxiliary information output by the preprocessor, runs in time $\Oh(n \log n)$, and outputs a compressed string $\tilde{U} \in \F_q^{\bH(Z) + \eps n}$. Further, for every auxiliary input, the map $Z \to \tilde{U}$ is a linear map.

\item {\sc Polar-Decompress} takes as input a Markov source $\HMM$ a compressed string $\tilde{U} \in \F_q^{\bH(Z)+\eps n}$ and the auxiliary information output by the preprocessor, runs in time  $\Oh(n^{3/2}\ell^2 + n \log{n})$ and outputs $\hat{Z} \in \F_q^n$.
\footnote{The runtime of the decompression algorithm can be improved to a runtime of $\Oh(n^{1+\delta}\ell^2 + n
\log{n})$ by a simple modification. In particular, by taking the input matrix $Z$
to be $n^{1-\delta} \x n^\delta$ instead of $n^{1/2} \x n^{1/2}$. In fact we believe the decoding algorithm can be improved to an $O(n \log n)$ time algorithm with some extra bookkeeping though we don't do so here.}
\end{enumerate}


\noindent The guarantee provided by the above algorithms is that
with probability at least $1-\exp(-\Omega(n))$, the Preprocessing Algorithm outputs
auxiliary information $S$ such that
$$\Pr_{Z \sim \HMM_n}[\textsc{Polar-Decompress}(\HMM, S; \textsc{Polar-Compress}(Z; S)) \neq Z] \leq \Oh(\frac{1}{n^2}),$$
provided $n > p(\tau/\eps)$ where $\tau$ is the mixing time of $\HMM$.

(In the above $\Oh(\cdot)$ hides constants depending $k$ and $q$, but not on $\ell$ or $n$.)
\end{theorem}
%

The above linear compression directly yields channel coding for additive Markov
channels, via a standard reduction (the details of which are in
Section~\ref{sec:proofs}.)
\begin{theorem}
\label{thm:main-channel}
For every prime $q$ and mixing matrix $M \in \F_q^{k \x k}$ there exists a
randomized preprocessing algorithm \textsc{Preprocess},
an encoding algorithm \textsc{Enc},
a decoding algorithm \textsc{Dec},
and a polynomial $p(\cdot)$ such that for every $\eps > 0$, the following properties hold:
\begin{enumerate}
\item {\sc Preprocess} is a randomized algorithm
that takes as input an additive Markov channel $\Ch_{\HMM}$ described by
Markov source $\HMM$ with $\ell$ states,
and $t \in \mathbb{N}$, and runs in time $\poly(n, \ell, 1/\eps)$ where
$n=k^{2t}$, and outputs auxiliary information for $\HMM_n$.
\item {\sc Enc} takes as input a message $x \in \F_q^r$, where $r \geq n(1 - \frac{\bH(Z)}{n} - \eps)$, as well as auxiliary information from the preprocessor and outputs
and computes {\sc Enc}$(x) \in \F_q^n$ in $\Oh(n \log n)$ time.
\item {\sc Dec} takes as input the Markov source $\HMM$, auxiliary information from the preprocessor
and a string $z \in \F_q^n$, runs in time $\Oh_q(n^{3/2}\ell^2 + n \log{n})$,
and outputs an estimate $\hat{x} \in \F_q^r$ of the message $x$.
\footnote{This can similarly be improved to a runtime of $\Oh_q(n^{1+\delta}\ell^2 + n
\log{n})$.}
\end{enumerate}
The guarantee provided by the above algorithms is that
with probability at least $1 - \exp(-\Omega(n))$, the Preprocessing algorithm outputs
$S$ such that for all $x \in \F_q^r$ we have
$$\Pr_{\Ch_\HMM}[\textsc{Dec}(\HMM; \Ch_{\HMM}(\textsc{Enc}(C; x))) \neq x]
\leq \Oh(\frac{1}{n^2}), $$
provided $n > p(\tau/\eps)$ where $\tau$ is the mixing time of $\HMM$.

(In the above $\Oh(\cdot)$ hides constants that may depend on $k$ and $q$ but not on $\ell$ or $n$.)
\end{theorem}

Theorem~\ref{thm:main-channel} follows relatively easily from Theorem~\ref{thm:main} and so in the next section we focus on the overview of the proof of the latter.




\newcommand{\oddZ}{Z_{\mathrm{odd}}}
\newcommand{\evenZ}{Z_{\mathrm{even}}}

\section{Overview of our construction}
\label{sec:overview}

\smallskip \noindent \textbf{Basics of polarization.}
We start with the basics of polarization in the setting of compressing samples from an i.i.d. source. To compress a sequence $Z \in \F_2^n$ drawn from some source, the idea is to build an invertible linear function $P$ such that for all but $\epsilon$ fraction of the output coordinates $i\in [n]$, the conditional entropy $H(P(Z)_i | P(Z)_{<i})$ is close to $0$ and or close to $1$. (Such an effect is called \emph{polarization}, as the entropies are driven to polarize toward the two extreme values.) Since a deterministic invertible transformation preserves the total entropy, it follows that roughly $H(Z)$ output coordinates can have entropy close to $1$ and $n-H(Z)$ coordinates have (conditional) entropy close to $0$. Letting $S$ denote the coordinates whose conditional entropies that are not close to zero, the compression function is simply $Z \mapsto P(Z)_S$, the projection of the output $P(Z)$ onto the coordinates in $S$.

Picking a random linear function $P$ would satisfy the properties above with high probability, but this is not known (and unlikely) to be accompanied by efficient algorithms. To get the algorithmics (how to compute $P$ efficiently, to determine $S$ efficiently, and to decompress efficiently) one uses a recursive construction of $P$. For our purposes the following explanation works best: Let $n=m^2$ and view $Z = (Z_{11},Z_{12},\ldots,Z_{mm})$ and as an $m \times m$ matrix over $\F_2$, where the elements of $Z$ arrive one row at a time. Let $P_m^{\rm row}(\cdot)$ denote the operation mapping $\F_2^{m \times m}$ to $\F_2^{m \times m}$ that applies $P_m$ to each row of separately. Let $P_m^{\rm column}(\cdot)$ denote the operation that applies $P_m$ to each column separately. Then $P_n(Z) = P_m^{\rm column}(P_m^{\rm row}(Z))^T$. The base case is given by $P_2(U,V) = (U+V,V)$.

Intuitively, when the elements of $Z$ are {\em independent} and identical, the operation $P_m$ already polarizes the outputs somewhat and so
a moderate fraction of the outputs of $P_m^{\rm row}(Z)$ have conditional entropies moderately close to $0$ or $1$. The further application of $P_m^{\rm column}(\cdot)$ further polarizes the output bringing a larger fraction of he conditional entropies of the output even closer to $0$ or $1$.

\medskip\noindent \textbf{Polarization for Markovian Sources.}
When applied to source $Z$ with memory, roughly the analysis in \cite{SasogluTal}, reinterpreted to facilitate our subsequent modification of the above polar constructuion, goes as follows: Since the elements of the row $Z_i$ are not really independent one cannot count on the polarization effects of $P_m^{\rm row}$. But, letting $U = P_m^{\rm row}(Z)$ one can show that most elements of the column of $U^j$ are almost independent of each other, provided $m$ is much larger than the mixing time of the source. (Here we imagine that the entries of $Z$ arrive row-by-row, so that the source outputs within each row are temporally well-separated from most entries of the previous row, when $m$ is large.) Further, this almost independence holds even when conditioning on the columns $U^{<j}$ for most values of $j$. Thus the operation $P_m^{\rm column}(\cdot)$ continues to have its polarization effects and this is good enough to get a qualitatively strong polarization theorem (about the operator $P_n$!).

The above analysis is asymptotic, proving that in the limit of $n \to \infty$, we get optimal compression. However, we do not know how to give an effective finite-length analysis of the polarization process for Markovian process, as the analysis in \cite{GX15,HAU14} crucially rely on independence which we lack within a row.

\medskip\noindent\textbf{Our Modified Code and Ingredients of Analysis.}
To enable a finite-length analysis, we make a minor, but quite important, alteration to the polar code: Instead of using $P_n(Z) = P_m^{\rm column}(P_m^{\rm row}(Z))^T$ we simply use the transformation $\tilde{P}_n = P_m^{\rm column}(Z)^T$ (or in other words, we replace the inner function $P_m^{\rm row}(\cdot)$ in the definition of $P_n$ by the identity function). 
This implies that we lose whatever polarization effects of $P_m^{\rm row}$ we may have been counting on, but as pointed out above, for Markov sources, we weren't counting on polarization here anyway! 

The crucial property we identify and exploit in the analysis is the following: the Markovian nature of the source plus the row-by-row arrival ordering of $Z$, implies that the distribution of the $j$'th source column $Z^j$ conditioned on the previous columns $Z^{<j} = z^{<j}$, is a \emph{close to} a product distribution, for all but the last few (say $\epsilon m$) columns. \footnote{We handle the non-independence in the last few columns, by simply outputting those columns $P_m(Z^j)$ in entirety, rather than only a set $S_j$ of entropy-carrying positions. This only adds an $\epsilon$ fraction to the output length, which we can afford.}

 It turns out that the analysis of the polar transform $P_m$ only needs independent inputs, which however need not be identically distributed. We are then able to apply the recent analysis from \cite{BGNRS}, essentially as black box, to argue that $P_m$ will compress each of the conditioned sources $Z^j | {Z^{< j} = z^{< j}}$ to its respective entropy, and also establish fast convergence via quantitatively strong polynomial (in the gap to capacity) upper bounds on the $m$ needed to achieve this. Further, we automatically benefit from the generality of the analysis in \cite{BGNRS}, which applies not only to the $2 \times 2$ transform $P_2$ at the base case, but in fact any $k \times k$ transform (satisfying some minimal necessary conditions) over an arbitrary prime field $\F_q$. Previous works on polar coding for Markovian sources~\cite{Sasoglu,SasogluTal,WHYLH} only applied for Boolean sources.

We remark that the use of the identity transform for the rows in $\tilde{P}_n$ is quite
counterintuitive. It implies that the compression matrix is a block diagonal
matrix (after some permutation of the rows and columns) --- and in turn this seems
to suggest that we are compressing different parts of the input sequence
``independently''. However this is not quite true. The relationship between the
blocks ends up influencing the final set $S$ of the bits of $\tilde{P}_n(Z)$
that are output by the compression algorithm. Furthermore the decompression
relies on the information obtained from the decompression of the blocks
corresponding to $Z^{<j}$ to compute the block $Z^j$. 

\medskip\noindent\textbf{Decompression algorithm.} Our alteration to apply the identity transform for the rows also helps us with the task of decompression. Toward this, we build on a decompression algorithm for memoryless sources from \cite{BGNRS} that is somewhat different looking from the usual ones in the polar coding literature.  This algorithm aims to compute $U = P_m^{\rm row}(Z)$ one {\em column} at a time, given $P_n(Z)|_S$. Given the first $j-1$ columns $U^{<j} = u^{<j}$, the algorithm first computes the conditional distribution of $U^j$ conditioned on $U^{<j} = u^{<j}$ and then uses a recursive decoding algorithm for $P_m$ to determine $U^j$. The key to the recursive use is again that the decoding algorithm works as long as the input variables are independent (and in particular, does not need them to be identically distributed).

In our Markovian setting, we now have to compute the conditional distribution of $Z^j$ conditioned on $Z^{<j} = z^{<j}$. But as mentioned above, this conditional distribution is close to a product distribution, say $D_j(z^{< j})$ (except for the last few columns $j$ where decompression is trivial as we output the entire column). Further, the marginals of this product distribution are easily computed using dynamic programming (via what is called the ``Forward Algorithm'' for hidden Markov models, described for completeness in Appendix~\ref{app:forward-algo}). We can then determine the $j$'th column $Z^j$ (having already recovered the first $j-1$ columns as $z^{< j}$) by running (in a black box fashion) the polar decompressor from \cite{BGNRS} for the memoryless case, feeding this product distribution $D_j(z^{< j})$ as the source distribution.





\medskip\noindent\textbf{Computing the output indices.}
Finally we need one more piece to make the result fully constructive. This is
the preprocessing needed to compute the subset $S$ of the coordinates of
$\tilde{P}_n(Z)$ that have noticeable conditional entropy. For the memoryless
case these computations were shown to be polynomial time computable in the works
of \cite{PHTT,GX15,TalVardy}. We manage to extend the ideas from Guruswami and Xia~\cite{GX15}
to the case of Markovian channels as well. It turns out the only ingredients
needed to make this computation work are, again, the ability to compute the
distributions of $Z^j$ conditioned on $Z^{<j} = z^{<j}$ for typical values of
$z^{<j}$. We note that unlike in the setting of memoryless channels (or i.i.d. sources) our preprocessing step is randomized. We believe this is related to the issue that there is no ``closed'' form solutions to basic questions related to Markovian sources and channels (such as the capacity of the channel in Figure~\ref{fig:one}) and this forces us to use some random sampling and estimation to compute some of the conditional entropies needed by our algorithms. 





\medskip\noindent\textbf{Organization of rest of the paper.} In the next section (Section~\ref{sec:construct}) we describe our compression and decompression algorithms. In Section~\ref{sec:analysis} we describe a notion of ``nice''-ness for the preprocessing stage and show that if the preprocessing algorithm returns a nice output, then the compression and decompression algorithm work correctly with moderately high probability (over the message produced by the source). 
In Section~\ref{sec:preprocess} we describe our preprocessing algorithm that returns a nice set with all but exponentially small failure probability (over its internal coin tosses). 
Finally in Section~\ref{sec:proofs} we give the formal proofs of Theorems~\ref{thm:main}~and~\ref{thm:main-channel}.

\section{Construction}
\label{sec:construct}

\subsection{Compression Algorithm}

Our compression, decompression and preprocessing algorithms are defined with respect to arbitrary mixing matrices $M \in \F_q^{k \x k}$. (Recall that mixing matrices were defined in  Definition~\ref{def:mixing}.) Though a reader seeking simplicity may set $k=2$ and $M = \left[\begin{array}{ll} 1 & 1\\ 0 & 1\end{array}\right]$.
Given integer $t$, let $m = k^t$ and let $P_m = P_{m,M}: \F_q^m \to \F_q^m$ be the polarization transform given by $P_m = M^{\otimes t}$.

\begin{algorithm}[H]
\caption{{\sc Polar-Compress}}
    \label{algo:compressor}
\begin{algorithmic}[1]
    \Const{$M \in \F_q^{k \x k}$, $m = k^t, n = m^2$}
    \Require{$Z = (Z_{11},Z_{12},\ldots,Z_{mm}) \in \F_q^n$, and sets $S_j \subseteq [m]$
    for $j \in [m]$}
    \Ensure{$U^j_{S_j} \in \F_q^{s_j}$ for all $j \in [m]$}
    \Comment{$s_j := |S_j|$ for $j \leq (1-\eps) m$, and $s_j := m$ otherwise.}
\Procedure{Polar-Compress}{$Z; \{S_j\}_{j \in [m]}$}
    \ForAll{$j \in [m]$}
        \State{Compute $U^j := P_m(Z^j)$.}
        \If{$j \leq (1-\eps)m$}
            \State{Output $U^j_{S_j}$}
        \Else
            \State{Output $U^j$}
        \EndIf
    \EndFor
\EndProcedure
\end{algorithmic}
\end{algorithm}

\subsection{Fast Decompressor}

The decompressor below makes black-box use of the \textsc{Fast-Decoder} from
\cite[Algorithm 4]{BGNRS}.

The \textsc{Fast-Decoder} takes as input the description of a \emph{product distribution}
$\cD_Z$ on inputs in $\F_q^m$, as well as the specified coordinates of the
compression $U$.
It is intended to decode from the encoding $U' \in \{\F_q
\cup \{\bot\}\}^m$, where $U := M^{\otimes t} Z$, coordinates of $Z$ are independent,
and $U'$ is defined by $U$ on the high-entropy coordinates of $U$ (and $\bot$
otherwise).
It outputs an estimate $\hat Z$ of the input $Z$.


\newcommand{\Dhj}{\mathcal{D}_{z^j | \hat{z}^{<j}}}
\begin{algorithm}[H]
\caption{{\sc Polar-Decompress}}
    \label{algo:fast-decompressor}
\begin{algorithmic}[1]
    \Const{$M \in \F_q^{k \x k}$, $m = k^t, n = m^2$}
    \Require{Markov Source $\HMM$ and $U^1_{S_1}, U^2_{S_2}, \dots, U^m_{S_m} \in \F_q^m$}
    \Ensure{$\hat Z \in \F_q^{m \x m}$}
\Procedure{Polar-Decompress}{$\HMM; U^1_{S_1}, U^2_{S_2}, \dots, U^m_{S_m}$}
    \ForAll{$j \in [m]$}
        \If{$j \leq (1-\eps)m$}
            \State{Compute the distribution
            $\Dhj \equiv \Ziid^j | \{\Ziid^{<j} = \hat Z^{<j}\}$}, using
            the \emph{Forward Algorithm} on Markov Source $\HMM$.
            \label{ln:inference}
            \State{Define $U^j \in \{\F_q \cup \{\bot\}\}^m$ by extending $U^j_{S_j}$
            using $\bot$ in the unspecified coordinates.}
            \State{Set $\hat Z^j \gets \text{\sc Fast-Decoder}(\Dhj; U^{j})$ } \label{line:fastdc}
            \label{ln:decoder}
        \Else
            \State{Set $\hat Z^j \gets (M^{-1})^{\otimes t} \hat U_{S_j}^j$} \Comment{Note here $S_j = [m]$}
        \EndIf
    \EndFor
    \State \Return{$\hat Z$}
\EndProcedure
\end{algorithmic}
\end{algorithm}

Note that, for a Markov source $\HMM$ on $\ell$ states, Line ~\ref{ln:inference} takes time
$\Oh_q(m^2\ell^2)$ (time $\Oh_q(m \ell^2)$ per coordinate of $\Ziid^j$, using the Forward
Algorithm).
The \textsc{Fast-Decoder} call in Line ~\ref{ln:decoder} takes time
$\Oh_q(m \log m)$.
Thus, the total runtime is $\Oh_q(m^3 \ell^2 + m^2 \log m) = \Oh_q(n^{3/2}
\ell^2 + n \log n)$.


\section{Analysis}
\label{sec:analysis}

\newcommand{\HZ}{\HMM_{m^2}}
\newcommand{\HZiid}{\HMM_m^{\otimes m}}

The goal of this section is to prove that the decompressor works correctly, with
high probablity, provided the preprocessing stage returns the appropriate sets
$\{S_j\}$.
Specifically, we prove Theorem~\ref{thm:main-analysis} as stated below. But first we need a definition of ``nice'' sets $\{S_j\}$: We will later show that pre-processing produces such sets and compression and decompression work correctly (w.h.p.) given nice sets.

\begin{definition}[$(\eps, \zeta)$-niceness]
\label{def:nice}
Let $\HMM$ be a Markov source.
For every $m \in \N$ and $n = m^2$,
let $\Ziid \sim \HZiid$ be the corresponding ``independent'' distribution.
Let $\bU := P_m^{\rm column}(\Ziid)$.

We call sets $S_1, S_2, \dots S_m \subseteq [m]$
\emph{``$(\eps, \zeta)$-nice''}
if they satisfy the following:
	\begin{enumerate}
		\item $\sum_j |S_j| \leq \bH(Z) + \eps n$
		\item $\forall j \in [m], i \not\in S_j:  \bH(\bU_{(i, j)} | \bU_{\prec (i, j)})
        < \zeta$
	\end{enumerate}
\end{definition}

Now, the rest of this section will show the following.

\begin{theorem}
	\label{thm:main-analysis}
	There exists a polynomial $p(\cdot)$ such that
	for every $\eps > 0$, $\tau > 0$, and $n =  m^2 > p(\tau/\eps)$ the following holds:

    Let $\HMM$ be an aperiodic irreducible Markov source with alphabet $\F_q$,
    mixing time
	$\tau$ and underlying state space $[\ell]$.
    Define random variables $Z = (Z_{11}, Z_{12} \dots Z_{mm}) \sim \HZ$
	as generated by $\HMM$.
	Then, for all sets $S_1, S_2, \dots S_m \subseteq [m]$ that are $(\eps,
    \zeta)$-nice as per Definition~\ref{def:nice}, we have:
	$$\Pr_{Z}[\textsc{Polar-Decompress}(\textsc{Polar-Compress}(Z;
	\{S_j\}_{j \in [m]})) \neq Z] \leq n\zeta + m \exp(-\eps m / \tau)$$

\end{theorem}

\subsection{Proof Overview}
Throughout this section, let $\HMM$ be a
stationary Markov source with alphabet $\F_q$ and mixing-time $\tau$.
The key part of the analysis is showing that compression and decompression
succeed when applied to the ``independent'' distribution $\Ziid \sim \HZiid$.
To do this, we first show that the compression transform ``polarizes''
entropies, which follows directly from the results of \cite{BGNRS,polar-small-error}.
Then we show that, provided ``nice'' sets can be computed
(low-entropy sets, a la Definition~\ref{def:nice}), the compression and decompression
succeed with high probability.
This also follows essentially in a black-box fashion from the results of
\cite{BGNRS}.
Finally, we argue that the compression and decompression also work for the
actual distribution $Z \sim \HZ$, simply by observing that the involved
variables are close in distribution.

We later describe how such ``nice'' sets can be computed in polynomial time,
given the description of the Markov source $\HMM$.

\subsection{Polarization}
In this section, we show that the compression transform $P^{\rm column}_m$
polarizes entropies.

\begin{lemma}
\label{lem:polarization}
Let $\HMM$ be a Markov source, and let $\Ziid \sim \HZiid$.
Let $\bU = P^{\rm column}_m(\Ziid)$.

    Then, there exists a polynomial $p(\cdot)$ such that for every $\eps > 0$, there exists $\beta > 0$
    such that if $m > p(1/\eps)$, the following holds:
    For all but $\eps$-fraction of indices $i, j \in [m] \times [m]$,
    the normalized entropy
    $$\bH(\bU_{i, j} | \bU_{\prec (i, j)}) \not\in (\exp(-m^\beta), 1- \eps)$$
\end{lemma}
\begin{proof}
We will show that for each column $\bU^j$, all but $\eps$-fraction of
indices $i \in [m]$ have entropies
$$\bH(\bU^j_i | \bU^j_{< i}, \bU^{< j}) \not\in (\exp(-m^\beta), 1- \eps)$$

Indeed, this follows directly from the analysis in \cite{polar-small-error}.
For each $j$, the set of variables
$
(\Ziid^j_1, \Ziid^{<j}_1),
(\Ziid^j_2, \Ziid^{<j}_2),
\dots,
(\Ziid^j_m, \Ziid^{<j}_m)
$
are
independent and identically distributed.
Thus, Theorem~\ref{thm:exp-polar} from \cite{polar-small-error} (reproduced below)
implies that the conditional entropies are polarized.
Specifically, let $p(\cdot)$ and $\beta$ be as guaranteed by
Theorem~\ref{thm:exp-polar}, for the distribution
$\mathcal{D}) \equiv (\Ziid^j_1, \Ziid^{<j}_1)$.
Then, since $P_m = M^{\otimes t}$, we have
\begin{align*}
\bH(\bU^j_i | \bU^j_{< i}, \bU^{< j})
&= \bH(\bU^j_i | \bU^j_{< i}, P_m(\Ziid^{< j})) \tag{by definition} \\
&= \bH(\bU^j_i | \bU^j_{< i}, \Ziid^{< j})  \tag{$P_m$ is invertible} \\
&\not\in (\exp(-m^\beta), 1- \eps) \tag{Theorem~\ref{thm:exp-polar} \qedhere}
\end{align*}
\end{proof}

The following theorem is direct from the works \cite{polar-small-error}.
\begin{theorem}
\label{thm:exp-polar}
For every $k \in \N$, prime $q$, mixing-matrix $M \in \F_q^{k \x k}$,
discrete set $\cY$, and any distribution $\mathcal{D} \in \Delta(\F_q \x
\cY)$, the following holds.
Define the random vectors $A := (A_1, A_2, \dots A_n)$
and $B := (B_1, B_2, \dots B_n)$ where $n=k^t$ and each component
$(A_i, B_i)$ is independent and identically distributed $(A_i, B_i) \sim \mathcal{D}$.

Let $X := M^{\otimes t}A$.
Then, the conditional entropies of $X$ are
\emph{polarized}:
There exists a polynomial $p(\cdot)$ and $\beta > 0$ such that for every $\eps > 0$,
if $n = k^t > p(1/\eps)$,
then all but $\eps$-fraction of indices $i \in [n]$ have normalized entropy
$$
\bH(X_i | X_{< i}, B)
\not\in (\exp(-n^\beta), 1- \eps) \ . $$
\end{theorem}

\subsection{Independent Analysis}
Now we show that the Polar Compressor and Decompresser succeed with high
probability, when applied to the
``independent'' input distribution $\Ziid$.

First, we recall the (inefficient) Successive-Cancellation Decoder of Polar
Codes. This is reproduced as in \cite{BGNRS}, with minor notational changes.
We will use this decoder to reason about the efficient fast decoder.

The \textsc{SC-Decoder} is intended to decode from the encoding $U_S$
where $U := M^{\otimes t} Z$, coordinates of $Z$ are independent, and $U_S$ is the
high-entropy coordinates of $U$.
It outputs an estimate $\hat U$ of $U$ that is correct with high probability,
from which we can decode the original inputs $\hat Z := (M^{-1})^{\otimes t}
\hat U$.

The \textsc{SC-Decoder} takes as input the \emph{product distribution} on inputs $\cD_Z
\in \Delta(\F_q^m)$, as well as the high-entropy coordinates $U_S$.

\begin{algorithm}[H]
\caption{Successive-Cancellation Decoder}
    \label{algo:sc-decoder}
\begin{algorithmic}[1]
    \Const{$M \in \F_q^{k \x k}, m = k^t, S \subseteq [m]$}
    \Require{Product distribution $\mathcal{D}_Z \in (\Delta(\F_q))^{\otimes m}$,
    and encoding $U_S$}
    \Ensure{${\hat U} \in \F_q^m$}
    \Procedure{SC-Decoder}{$\mathcal{D}_Z; U_S$}
    \State Compute joint distribution $\mathcal{D_U} \in \Delta(\F_q^m)$ of
    $\{U := M^{\otimes t} Z\}_{Z \sim \cD_Z}$
    \ForAll{$i = 1, \dots, m$}
    \If{$i \not\in S$}
        \State $\hat{{U}}_i \gets \argmax_{x \in \F_q} \Pr_{U \sim \cD_U}({U}_i = x)$
    \Else
        \State $\hat{{U}}_i \gets U_i$
    \EndIf
    \State Update distribution $\cD_U$ to be conditioned on ${U}_i = \hat{{U}}_i$
    \EndFor
    \State \Return ${\hat U}$
\EndProcedure
\end{algorithmic}
\end{algorithm}
Note that several of the above steps, including computing the joint distribution
$\cD_U$ and marginal distributions of $U_i$, are not computationally efficient.

The following claim is equivalent to \cite[Claim A.1]{BGNRS}, and states that the
failure probability of the \textsc{SC-Decoder} is at most the sum of conditional
entropies on the unspecified coordinates of $U$.
\begin{claim}
\label{lem:sc-dec-H}
    Let $\td{Z} \in \F_q^m$ be a random vector with independent
    (not necessarily identically distributed) components $\td{Z}_i$.
    Denote the distribution of $\td Z$ as $\mathcal{D}_{\td Z}$.
    Let $\td U := P_m(\td Z)$, and $S \subseteq [m]$.

    Then,
    $$\Pr[\text{\sc SC-Decoder}(\mathcal{D}_{\td Z}; \td{U}_S) \neq U]
    \leq \sum_{i \not\in S}
    \bH(\td{U}_i | \td{U}_{< i})
    $$
\end{claim}

\newcommand{\Dj}{\mathcal{D}_{z^j | z^{<j}}}
\begin{claim}
\label{lem:sc-dec-matrix}
Let $\Ziid \sim \HZiid$, and let $\bU = P^{\rm column}_m(\Ziid)$.
For a fixed $j \in [m]$
and fixed conditioning $z^{< j} \in \F_q^{m \x (j-1)}$,
let $\Dj$ denote the distribution $\Ziid^{j} | \{\Ziid^{< j} = z^{< j}\}$.

Then, for all $j \in [m]$ and all $S \subseteq [m]$,
    $$\Pr_{\substack{
    z \sim \Ziid\\
    u \gets P_m(z^j)\\
    \Dj \equiv \{\Ziid^{j} | \Ziid^{< j} = z^{< j}\}
    }}
    [\text{\sc SC-Decoder}(\Dj; u^{j}_S) \neq u^{j}]
    \leq \sum_{i \not\in S}
    \bH(\bU^j_i | \bU^j_{< i}, \bU^{< j})
    $$
\end{claim}
\begin{proof}
This follows directly from Claim~\ref{lem:sc-dec-H}.
\begin{align*}
    &\E_{\substack{
    z \sim \Ziid\\
    u \gets P_m(z^j)\\
    \Dj \equiv \{\Ziid^{j} | \Ziid^{< j} = z^{< j}\}
    }}
    [\1\{\text{\sc SC-Decoder}(\Dj; u^{j}_S) \neq u^{j}\}]\\
    &=
    \E_{\substack{
    z^{< j} \sim \Ziid^{< j}\\
    }}[
    \E_{\substack{
    \Dj \equiv \{\Ziid^{j} | \Ziid^{< j} = z^{< j}\}\\
    z^j \sim \Dj\\
    u \gets P_m(z^j)
    }}[
    \1\{\text{\sc SC-Decoder}(\Dj; u^{j}_S) \neq u^{j}\}
    ]]\\
    &=
    \E_{\substack{
    z^{< j} \sim \Ziid^{< j}\\
    }}\left[
    \Pr_{\substack{\td{Z} \sim \Dj\\ \td{U} \gets P_m(\td Z)}}[
    \text{\sc SC-Decoder}(\cD_{\td Z}; \td{U}_S) \neq \td{U}
    ]\right]\\
    &\leq
    \E_{\substack{
    z^{< j} \sim \Ziid^{< j}
    }}\left[
    \sum_{i \not\in S} \bH(\bU^j_i | \bU^j_{< i}, \Ziid^{< j} = z^{< j})
    \right] \tag{by Claim~\ref{lem:sc-dec-H}}\\
    &= \E_{\substack{
    z^{< j} \sim \Ziid^{< j}\\
    u^{<j} \gets P_m(z^{<j})
    }}\left[
    \sum_{i \not\in S} \bH(\bU^j_i | \bU^j_{< i}, \bU^{< j} = u^{< j})
    \right]\\
    &=
    \sum_{i \not\in S} \bH(\bU^j_i | \bU^j_{< i}, \bU^{< j}) \qedhere
\end{align*}
\end{proof}

\noindent
Using the \textsc{SC-Decoder}, we can define the following (inefficient)
decompresser.
We will then relate its performance to the fast decompressor, and thereby conclude the desired correctness property of the latter.

\begin{algorithm}[H]
\caption{SC Polar Decompressor}
    \label{algo:sc-decompressor}
\begin{algorithmic}[1]
    \Const{$M \in \F_q^{k \x k}$, $m = k^t, n = m^2$}
    \Require{$U^1_{S_1}, U^2_{S_2}, \dots, U^m_{S_m} \in \F_q^m$, and Markov
    source $\HMM$}
    \Ensure{$\hat Z \in \F_q^{m \x m}$}
\Procedure{SC-Polar-Decompress$_{\HMM}$}{$U^1_{S_1}, U^2_{S_2}, \dots, U^m_{S_m}$}
    \ForAll{$j \in [m]$}
        \If{$j \leq (1-\eps)m$}
            \State{Compute the distribution
            $\Dhj \equiv \Ziid^j | \{\Ziid^{<j} = \hat z^{<j}\}$}, for $\Ziid
            \sim \HZiid$.
            \State{Set $\hat U^j \gets \text{\sc SC-Decoder}(\Dhj; U^{j}_{S_j})$ } \label{ln:Uhat}
        \Else
            \State{Set $\hat U^j \gets U^j_{S_j}$} \Comment{Note here $S_j = [m]$} \label{ln:Uhat2}
        \EndIf
        \State{Set $\hat Z^j \gets (M^{-1})^{\otimes t} \hat U^j$} \label{ln:hatZ}
    \EndFor
    \State \Return{$\hat Z$}
\EndProcedure
\end{algorithmic}
\end{algorithm}

\begin{claim}
\label{lem:sc-entropy-bound}
Let $\Ziid \sim \HZiid$, and $\bU := P^{\rm column}_m(\Ziid)$.
Then, for all sets $S_1, S_2, \dots S_m \subseteq [m]$,
$$\Pr_{\substack{\bU \gets P^{\rm column}_m(\Ziid)}}
[\text{\sc SC-Polar-Decompress$_{\HMM}$}(U^1_{S_1}, U^2_{S_2}, \dots, U^m_{S_m}) \neq
\bU]
\leq \sum_{j\in [m], i \not\in S_j}
\bH(\bU^j_i | \bU^j_{< i}, \bU^{< j})
$$
\end{claim}
\begin{proof}
This follows directly from Claim~\ref{lem:sc-dec-matrix}.
\begin{align*}
&\Pr_{\substack{\bU \gets P^{\rm column}_m(\Ziid)}}
[\text{\sc SC-Polar-Decompress}(U^1_{S_1}, U^2_{S_2}, \dots, U^m_{S_m}) \neq
\bU]\\
&=
\Pr_{\substack{\bU \gets P^{\rm column}_m(\Ziid)}}
\left[
\bigcup_{j \in [m]} \{ \hat U^j \neq \bU^j \text{ and } \hat U^{<j} = \bU^{<
j}\}
\right] \tag{for random variables $\hat U^j$ defined as in
Lines~\ref{ln:Uhat},~\ref{ln:Uhat2} of
Algorithm~\ref{algo:sc-decompressor}}\\
&\leq
\sum_{j \in [m]}
\Pr[\hat U^j \neq \bU^j | \hat U^{<j} = \bU^{< j}]\\
&=
\sum_{j \in [m]}
\Pr[\hat U^j \neq \bU^j | \hat Z^{<j} = \Ziid^{< j}] \tag{by definition of $\hat
Z$ in Line~\ref{ln:hatZ}}\\
&=
\sum_{1 \leq j \leq (1-\eps)m}
\Pr_{\substack{
    z \sim \Ziid\\
    u \gets P_m(z^j)\\
    \Dj \equiv \{\Ziid^{j} | \Ziid^{< j} = z^{< j}\}
    }}
    [\text{\sc SC-Decoder}(\Dj; u^{j}_{S_j}) \neq u^{j}]\\
    &\leq \sum_{j \in [m]} \sum_{i \not\in S_j}
    \bH(\bU^j_i | \bU^j_{< i}, \bU^{< j}) \quad \text{(by Claim~\ref{lem:sc-dec-matrix})} \quad \qedhere
\end{align*}
\end{proof}

We now analyze the fast decompressor. This is defined identically to the SC
Polar Decompressor, except using the \textsc{Fast-Decoder} from
\cite[Algorithm 4]{BGNRS} instead of the Successive-Cancellation Decoder.
Note the \textsc{Fast-Decoder} outputs an estimate of the input $\hat Z$
directly.

The following claim reproduced from \cite[Lemma A.4]{BGNRS} states that the Fast
Decoder operates identically as the Successive-Cancellation Decoder.
\begin{claim}
    \label{lem:fast-equivalence}
    For all $m = k^t$, all product distributions $\cD_Z$ over $\F_q^m$
(ie, where each coordinate ${Z}_{{i}} \in \F_q$ is independent),
for all sets $S \subseteq [m]$
and all values $U \in (\F_q \cup \{\bot\})^{m}$
of the coordinates $S$ (such that $\forall i \in S: U_i \neq \bot$),
the following holds:
$$
M^{\otimes t}
\cdot
\textsc{Fast-Decoder}(\cD_Z; U)
=
\textsc{SC-Decoder}(\cD_Z; U_S).
$$
\end{claim}

In particular, this directly yields the following analogue of
Claim~\ref{lem:sc-entropy-bound}, for \textsc{Polar-Decompress}.

\begin{claim}
\label{lem:fast-entropy-bound}
Let $\Ziid \sim \HZiid$ and $\bU := P^{\rm column}_m(\Ziid)$.
Then, for all sets $S_1, S_2, \dots S_m \subseteq [m]$,
$$\Pr_{\substack{\bU \gets P^{\rm column}_m(\Ziid)}}
[\text{\sc Polar-Decompress}(U^1_{S_1}, U^2_{S_2}, \dots, U^m_{S_m}) \neq
\bU]
\leq \sum_{j\in [m], i \not\in S_j}
\bH(\bU^j_i | \bU^j_{< i}, \bU^{< j})
$$
\end{claim}

\subsection{Proof of Main Theorem}
At this point, we can show the entire process of compression and
decompression succeeds with high probability, proving
Theorem~\ref{thm:main-analysis}.

First, we argue $\HZ$ and $\HZiid$ are close in the appropriate sense.
\begin{lemma}
\label{lem:hybrids}
    Let $Z \sim \HZ$ and $\Ziid \sim \HZiid$.
    Then, for every $\ell \in [m]$, the distribution of $Z^{< m - \ell}$ and $\Ziid^{< m - \ell}$ are
    $m\cdot \exp(-\ell / \tau)$-close in $L_1$.
\end{lemma}
\begin{proof}
We proceed by a sequence of $m$ hybrids, changing one row at a time to being
independent.
    Let the $i$-th hybrid be $H_i := Z^{< m - \ell}_{\leq i} \circ \Ziid^{< m -
    \ell}_{> i}$,
    that is, the first $i$ rows of $Z$, with the remaining rows replaced by iid
    copies of $Z_1$.

Consider moving from $H_{i+1}$ to $H_i$.
Conditioned on the first $i$ rows of $Z^{< m - \ell}$, the distribution of the
hidden state of the Markov source, at the beginning of the $(i+1)$th row,
    is $\exp(-\ell/\tau)$-close to its stationary distribution $\pi$ (since
    $\ell$ steps pass between $Z_{i, m-\ell}$ and $Z_{i+1, 1}$).
Recall that the distribution of $Z_1$ is generated by the Markov source starting
from $\pi$.
    Thus, the distribution of the $(i+1)$th row of $Z$,
    conditioned on the first $i$ rows of $Z^{< m - \ell}$,
    is $\exp(-\ell/\tau)$-close to the distribution of $Z_1$.
    So, $|H_{i+1} - H_i|_1 \leq \exp(-\ell/\tau)$.
    Since we pass through $m$ hybrids, the total $L_1$ distance is at most
 $m\cdot \exp(-\ell / \tau)$.
\end{proof}

\begin{proof}[Proof of Theorem~\ref{thm:main-analysis}]
First, we show the corresponding claim about the ``independent'' distribution
$\Ziid \sim \HZiid$:

\begin{claim}
\label{lem:fullZiid}
For $\Ziid \sim \HZiid$, we have:
$$\Pr_{\bU \gets P^{\rm
column}_m(\Ziid)}[\textsc{Polar-Decompress}(\bU^1_{S_1}, \bU^1_{S_2}, \dots
\bU^m_{S_m}) \neq \Ziid] \leq n\zeta$$

or equivalently,
$$\Pr_{\Ziid \sim \HZiid}[\textsc{Polar-Decompress}(\textsc{Polar-Compress}(\Ziid;
\{S_j\}_{j \in [m]})) \neq \Ziid] \leq n \zeta$$
\end{claim}
\begin{proof}
By Claim~\ref{lem:fast-entropy-bound}, we have
\begin{align*}
\Pr_{\bU \gets P^{\rm column}_m(\Ziid)}[
\textsc{Polar-Decompress}(\bU^1_{S_1}, \bU^1_{S_2}, \dots
\bU^m_{S_m}) \neq \Ziid]
&\leq \sum_{j\in [m], i \not\in S_j}
\bH(\bU^j_i | \bU^j_{< i}, \bU^{< j})\\
&\leq n \zeta \quad \text{(by $(\eps, \zeta)$-niceness)} \quad \qedhere
\end{align*}
\end{proof}


Continuing the proof of Theorem~\ref{thm:main-analysis},
notice that the composition of \textsc{Polar-Compress} and \textsc{Polar-Decompress}
always operate as the identity transform on the inputs $Z^j$ for $j > (1-\eps)m$.
Thus, it suffices to consider the behavior of this composition on inputs
$Z^{\leq (1-\eps)m}$.
In this case, Lemma~\ref{lem:hybrids} guarantees that the distributions of
$\Ziid^{\leq (1-\eps)m}$ and $Z^{\leq (1-\eps)m}$ are close in $L_1$, and thus
we may conclude by Claim~\ref{lem:fullZiid}:

\begin{align*}
&\Pr_{Z \sim \HZ}[\textsc{Polar-Decompress}(\textsc{Polar-Compress}(Z;
\{S_j\}_{j \in [m]})) \neq Z]\\
&\leq
\Pr_{\Ziid \sim \HZiid}[\textsc{Polar-Decompress}(\textsc{Polar-Compress}(\Ziid;
\{S_j\}_{j \in [m]})) \neq \Ziid]
+ m \exp(-\eps m / \tau) \tag{Lemma~\ref{lem:hybrids}}\\
&\leq n\zeta + m \exp(-\eps m/ \tau) \quad \text{(Claim~\ref{lem:fullZiid})} \quad \qedhere
\end{align*}

\end{proof}

\section{Preprocessing}
\label{sec:preprocess}

\newcommand{\dc}{d_C}
\newcommand{\DD}{\mathbb{D}}

In this section, we describe a pre-processing algorithm to find the
$(\eps, \zeta)$-nice sets, as defined in Definition~\ref{def:nice}, that are required by the compression and decompression algorithms.
Recall the notion of a mixing matrix (Definition~\ref{def:mixing}). The following theorem shows that for every prime alphabet $\F_q$ and mixing matrix $M \in \F_q^{k \x k}$, there is an efficient algorithm that can find nice sets in polynomial time.
Specifically, we prove the following theorem.

\begin{theorem}
\label{thm:preproc}
For every prime $q$ and mixing-matrix $M \in \F_q^{k \x k}$,
there exists a polynomial $p(\cdot)$ and a polynomial time preprocessing algorithm
\textsc{Polar-Preprocess} (Algorithm~\ref{algo:subset}), such that
for every $\eps > 0$ and $m > p(1/\eps)$, the following holds:

Let $\HMM$ be a Markov source with mixing-time
$\tau$, alphabet $\F_q$, and underlying state space $[\ell]$.
Let $\Ziid \sim \HZiid$ for $m = k^t$, and $\bU := P_m^{\rm column}(\Ziid)$.

Let
$$S_1, S_2, \dots S_m \gets \textsc{Polar-Preprocess}(q,k,t, M, \HMM)$$

Then, except with probability $\exp(-\Omega(m))$ over the randomness of the
algorithm, the output sets $S_1, S_2, \dots S_m \subseteq [m]$
are $(\eps, \zeta = \Oh(\frac{1}{n^3}))$-nice for $\HMM$.
Further, the algorithm runs in time
$poly_q(m, \ell, 1/\eps)$.
\end{theorem}

Our main goal will be to estimate the conditional entropies
$$
\bH(\bU_{(i, j)} | \bU_{\prec (i, j)})
=
\bH(\bU_{i, j} | \bU_{<i, j}, \bU^{< j})
=
\bH(\bU_{i, j} | \bU_{<i, j}, \Ziid^{< j})
$$
for $\Ziid \sim \HZiid$ and $\bU := P^{\rm column}_m(\Ziid)$.
Then, we will construct the ``nice'' sets by defining, for each $j$,
$S_j$ as the set of indices with high entropy:
$S_j := \{i \in [m]: \bH(\bU_{i, j} | \bU_{<i, j}, \Ziid^{< j}) > \frac{1}{n^3}
\}$.
By Polarization (Lemma~\ref{lem:polarization}), these sets have size at most
$\sum_j |S_j| \leq \bH(Z) + \eps n$, since they must have conditional entropies close to $1$
(except possibly for some $\eps$ fraction of indices $(i, j) \in [m] \x [m]$).

We will estimate conditional entropies
$\bH(\bU_{i, j} | \bU_{<i, j}, \Ziid^{< j})$
by approximately tracking the
distribution of variables as we apply successive tensor-powers of $M$.
Since we are only interested in conditional entropies, it is sufficient to
``quantize'' the true distribution of, for example $U_i | U_{< i}$,
into an approximation $U_i | A$, such that
$H(U_i | U_{< i}) \approx H(U_i | A)$.
This algorithm follows the same high-level strategy of \cite{GX15}, of approximating
the conditional distributions via quantized bins.
It turns out that this strategy can be implemented for Markov sources, using the
fact that Markov sources are constructive.
We define our notions of approximation, and formalize this strategy below.

\subsection{Notation and Preliminaries}

\begin{definition}[Associated Conditional Distribution]
Let $X$ be a random variable taking values in universe $U$,
and let $W$ be an arbitrary random variable.
Let $\cD_{X | w} \in \Delta(U)$ denote the conditional distribution of $X | \{W = w\}$.
Let $\DD_{X | W} \in \Delta(\Delta(U))$ be the distribution over $\cD_{X | w}$ defined
by sampling $w \sim W$.
We call $\DD_{X | W}$ the \emph{associated conditional distribution to $X | W$}.
\end{definition}

As above, we use boldface $\DD$ to denote objects of type $\Delta(\Delta(U))$.
Note that we can operate on conditional distributions as we would on their
underlying random variables. For example, for random variables $(A_1, W)$ and
$(A_2, Y)$ such that $A_1, A_2 \in \F_q$ and
$(A_1, W)$ is independent from $(A_2, Y)$,
the associated conditional distribution of $A_1 + A_2 | Y, W$ can be computed
from the associated conditional distributions of $A_1 | Y$ and $A_2 | W$.
To more easily describe such operations on conditional distributions (which may
not always arise from underlying random variables), we define the \emph{implicit
random variables} associated to a conditional distribution:

\begin{definition}[Implicit Random Variables Associated to Conditional
Distribution]
For every $\DD_{X | W} \in \Delta(\Delta(U))$,
define \emph{implicit random variables $X, W$ associated to $\DD_{X | W}$}
as random variables $(X, W)$ such that the associated conditional distribution
to $X | W$ is exactly $\DD_{X | W}$. Note that there is not a unique choice of
such random variables.
\end{definition}
%
Using this, we can naturally define (for example) $\DD_{A_1 + A_2 | W, Y}$
and $\DD_{A_2 | W, Y, A_1 + A_2}$ from any
$\DD_{A_1, W}, \DD_{A_2, Y} \in \Delta(\Delta(\F_q))$.
Note that we will always be performing such operations assuming independence of
the involved implicit random variables, ie $(A_1, W)$ and $(A_2, Y)$.

\begin{definition}[Conditional Distance]
Let $(X, W)$ and $(Y, Z)$ be two joint distributions, such that $X$ and $Y$ take
values in the same universe $U$.
Let $\DD_{X | W}$ and $\DD_{Y | Z}$ be the associated distributions in
$\Delta(\Delta(U))$.
Then, define the \emph{conditional distance}
$$
\dc(\DD_{X | W}, \DD_{Y | Z})
:= \min_{\substack{(A, B): \text{a distribution in } \Delta(\Delta(U) \x \Delta(U))\\
\text{s.t. marginals of $A$ match $\DD_{X | W}$, and}\\
\text{marginals of $B$ match $\DD_{Y | Z}$}
}}
~~\E_{(D_A, D_B) \sim (A, B)}[||D_A - D_B||_1]
$$

Note that $\dc$ can be equivalently defined as an optimal transportation cost
between two distributions in $\Delta(\Delta(U))$, where the cost of moving a
unit of mass between points $D_i, D_j \in \Delta(U)$ is $||D_i - D_j||_1$.
\end{definition}

This metric behaves naturally under post-processing:
\begin{claim}
\label{lem:postproc}
For all $\DD_{X|W}, \DD_{X'|W'} \in \Delta(\Delta(U))$, and any $f: U \to V$,
$$
\dc(\DD_{f(X)|W}, \DD_{f(X')|W'})
\leq \dc(\DD_{X|W}, \DD_{X'|W'})
$$
\end{claim}

For computational purposes, we represent the space of distributions using $\eps$-nets:
\begin{definition}[$\eps$-nets]
For every set $U$ and any $\eps > 0$, let $T_\eps(U) \subseteq \Delta(U)$ be an $\eps$-net of
$\Delta(U)$ with respect to $L_1$.
That is, for every $\cD \in \Delta(U)$, there exists $\hat{\cD} \in T_\eps(U)$ such
that
$||\cD - \hat{\cD}||_1 \leq \eps$.

Note that for $|U| = |\F_q| = q$, $T_{\eps}(U)$ can be chosen such that
$|T_{\eps}(U)| \leq {\frac{q}{\eps} + q \choose q} \leq (\frac{2q}{\eps})^q =
poly_q(1/\eps)$.

Moreover, $\Delta(T_{\eps}(U))$ is an $\eps$-net of $\Delta(\Delta(U))$ under
the $\dc$-metric.
\end{definition}

\subsection{Conditional Distribution Approximation}
The below procedure takes as input a conditional distribution
$\DD_{Z | W} \in \Delta(\Delta(\F_q))$,
and computes an approximation to the conditional distribution
of $U_I | (U_{\prec I}, W_1, \dots W_{k^t})$, for an index $I \in [k]^t$,
where $U := M^{\otimes t} Z$ and $\{(Z_i, W_i)\}_{i \in [k^t]}$ are independently defined by $\DD_{Z|W}$.

\begin{algorithm}[H]
\caption{Conditional Distribution Approximation}
\label{algo:dist-approx}
\begin{algorithmic}[1]
\Require{Conditional distribution on inputs $\DD_{Z | W} \in \Delta(\Delta(\F_q))$,
$\eps > 0$, $t \in \N$, index $I \in [k]^t$, and $M \in \F_q^{k \x k}$}
\Ensure{Conditional distribution $\tilde{\DD}_{U | W} \in \Delta(\Delta(\F_q))$,
an approximation to $U_I | (U_{\prec I}, W_1, \dots W_{k^t})$ for
$U := M^{\otimes t} Z$ and $\{(Z_i, W_i)\}_{i \in [k^t]}$ independently defined by $\DD_{Z|W}$.}
\Procedure{ApproxDist}{$\DD_{Z|W}, \eps, t, I = (I_1, \dots, I_t), M$}
\If{$t=0$}
\State \Return $\DD_{Z|W}$
\Else
\State $\hat{\DD}_{Z| Y} \gets \textsc{ApproxDist}(\DD_{Z|W}, \eps/(2k), t-1, I_{< t}=
(I_1, \dots, I_{t-1}), M)$ \label{ln:induction}
\State $j \gets I_t$.
\State Explicitly compute the following conditional distribution
$\hat{\DD}_{U_j | U_{< j}, Y_1, \dots Y_k} \in \Delta(\Delta(\F_q))$:

\State \quad \quad Let $(Z,Y)$ be the implicit random variables associated
to $\hat{\DD}_{Z | Y}$.
\State \quad \quad Let $\{(Z_i, Y_i)\}_{i \in [k]}$ be independent random
variables distributed identically to $(Z, Y)$.
\State \quad \quad Define random vector $U := M \cdot Z'$, Where
$Z' = (Z_1, \dots Z_k)$.
\State \quad \quad Let $\hat{\DD}_{U_j | U_{< j}, Y_1, \dots Y_k}$ be the associated
conditional distribution to $U_j | U_{< j}, Y_1, \dots Y_k$.  \label{ln:Mtensor}

\State Round $\hat{\DD}_{U_j | U_{< j}, Y_1, \dots Y_k}$ to
$\tilde{\DD}_{U | Y} \in \Delta(T_{\eps/2}(\F_q))$,
a point in the $\eps/2$-net of $\Delta(\Delta(\F_q))$ under $\dc$. \label{ln:round}
\State \Return $\tilde{\DD}_{U | Y}$.
\EndIf
\EndProcedure
\end{algorithmic}
\end{algorithm}

Note that if the input $\DD_{Z | W}$ is specified in an $\eps$-net
$\Delta(T_{\eps}(\F_q))$, then the above procedure runs in time
$poly_q(m, 1/\eps)$ for $m = k^t$.

\begin{lemma}
\label{lem:approxdist}
For all $\DD_{Z|W} \in \Delta(\Delta(U)), \eps > 0, t \in \N, M \in \F_q^{k \x k}$,
and $I \in [k]^t$, we have
$$\dc(\textsc{ApproxDist}(\DD_{Z|W}, \eps, t, I, M)
~~,~~
\DD_{U_I | U_{\prec I}, W_1, \dots W_{k^t}}) \leq \eps$$
where $\DD_{U_I | U_{\prec I}, W_1, \dots W_{k^t}}$ is the associated conditional
distribution to the random variables defined as follows.
Let $(Z,W)$ be the implicit random variables associated
to $\DD_{Z | W}$.
Let $\{(Z_i, W_i)\}_{i \in [k^t]}$ be independent random
variables distributed identically to $(Z, W)$.
Finally, define random vector $U := M^{\otimes t} \cdot Z'$, where
$Z' = (Z_1, \dots Z_{k^t})$.
\end{lemma}
\begin{proof}
First, notice that for $\eps=0$ (ie, omitting the rounding in
Line~\ref{ln:round}), these distributions are identical:
$$\textsc{ApproxDist}(\DD_{Z|W}, \eps=0, t, I, M) \equiv
\DD_{U_I | U_{\prec I}, W_1, \dots W_k}$$
It remains to prove that the rounding approximately preserves this.
We can prove the lemma by induction on $t$.
Suppose the statement holds for $t-1$.
Let $I' := I_{< t}$.
Let $\hat{\DD}_{Z | Y}$ be the result of the recursive call on Line~\ref{ln:induction}.
By the inductive hypothesis,
$$\dc(\hat{\DD}_{Z | Y} ~,~
\DD_{U'_{I'} | U'_{\prec I'}, W_1, \dots W_{k^{t-1}}}) \leq \eps/(2k)$$
where $U' := M^{\otimes t-1} \tilde Z$ for $\tilde Z := (Z_1, \dots, Z_{k^{t-1}})$ and
$\{(Z_i, W_i)\}_{i \in [k^t]}$ independent random variables
distributed according to $\DD_{Z | W}$.

Now, by the triangle inequality and Claim~\ref{lem:postproc},
the distribution $\hat{\DD}_{U_I | U_{\prec I}, Y_1, \dots Y_k}$
(defined in Line~\ref{ln:Mtensor} using the tensor-product recursion)
satisfies:
$$
\dc(\hat{\DD}_{U_I | U_{\prec I}, Y_1, \dots Y_k}
~,~
\DD_{U_I | U_{\prec I}, W_1, \dots W_{k^t}})
\leq \eps/2
$$
Now, by rounding to an $\eps/2$-net in Line~\ref{ln:round}, this distance
is distorted by at most an additional $\eps/2$. Thus:
$$
\dc(\tilde{\DD}_{U | Y}
~,~
\DD_{U_I | U_{\prec I}, W_1, \dots W_{k^t}})
\leq \eps \ . \qedhere
$$
\end{proof}

\subsection{Approximating Conditional Entropies}
Here we use Algorithm~\ref{algo:dist-approx} directly to approximate conditional
entropies:

\begin{theorem}
\label{thm:entropy-approx}
For every field $\F_q$, conditional distribution $\DD_{Z | W} \in
\Delta(\Delta(\F_q))$, matrix $M \in \F^{k \x k}$,
$t \in \N, m = k^t$,
and $\gamma > 0$,
consider the random variable $U := M^{\otimes t} Z$ where each
$\{(Z_i, W_i)\}_{i \in [m]}$ is sampled independently from $\cD_{Z, W}$.

Then, Algorithm~\ref{algo:entropy-approx}
outputs
$\hat h_1, \dots \hat h_m \gets \textsc{ApproxEntropy}(\DD_{Z|W}, \gamma, t, M)$
such that
$$\forall i \in [m]: \hat h_i = \bH(U_i | U_{< i}, W_1, \dots, W_m) \pm \gamma$$

Further, if the input $\DD_{Z | W}$ is specified in an $\eps$-net
$\Delta(T_{\eps}(\F_q))$, then the above procedure runs in time
$poly_q(m, 1/\eps, 1/\gamma)$.
\end{theorem}

\begin{algorithm}[H]
\caption{Entropy Approximation}
\label{algo:entropy-approx}
\begin{algorithmic}[1]
\Require{$\gamma > 0$, $t \in \N$, Conditional distribution $\DD_{Z | W} \in
\Delta(\Delta(\F_q))$, and $M \in \F_q^{k \x k}$}
\Ensure{$\{ \hat h_i \in \R\}_{i \in [k^t]}$}
\Procedure{ApproxEntropy}{$\DD_{Z|W}, \gamma, t, M$}
\State $m \gets k^t$
\State $\eps \gets \gamma^2$
\ForAll{$I \in [k]^t$}
\State $\DD_{U | Y} \gets \textsc{ApproxDist}(\DD_{Z|W}, \eps, t, I, M)$
\State $\hat h_I \gets \bH(U | Y)$, the conditional entropy of the implicit random
variables $(U, Y)$ associated to $\DD_{U | Y}$.
\EndFor
\State \Return $\{\hat h_i\}_{i \in [k^t]}$ \Comment{Abusing notation by
identifying $[k]^t$ with $[k^t]$.}
\EndProcedure
\end{algorithmic}
\end{algorithm}

\begin{proof}[Proof of Theorem~\ref{thm:entropy-approx}]
Correctness of Algorithm~\ref{thm:entropy-approx} follows from the fact that
$\gamma^2$-closeness in the $\dc$-metric implies $\gamma$-closeness of (normalized)
conditional entropies, as in Lemma~\ref{lem:entropy-dc} below.
Thus, using Algorithm~\ref{algo:dist-approx} to approximate the conditional
distributions within $\gamma^2$ is sufficient.
\end{proof}

\begin{claim}
\label{lem:entropy-dc}
For any finite set $U$, consider any two conditional distributions
$\DD_{X | W}, \DD_{X' | W'} \in \Delta(\Delta(U))$.

Then, the conditional entropies of the associated random variables
satisfy
$$\dc(\DD_{X | W}, \DD_{X' | W'}) \leq \eps \implies
|H(X | W) - H(X' | W')| \leq \eps \log(\frac{|U|}{\eps}) \leq \eps^2 \log(|U|)$$
\end{claim}

The above claim in turn follows directly from the transportation-cost definition of
$\dc$, and Claim~\ref{lem:entropy-L1} below.

\begin{claim}
\label{lem:entropy-L1}
Let $X, Y$ be two random variables taking values in the same finite universe $\Sigma$.
Then, the \emph{non-normalized entropies} satisfy
$$||X - Y||_1 \leq \eps \implies |H(X) - H(Y)| \leq \eps
\log(\frac{|\Sigma|}{\eps})$$
\end{claim}
\begin{proof}
Let $N := |\Sigma|$, and assume without loss of generality that the random
variables take values in $[N]$.
Let $p_i := \Pr[X = i]$ and $q_i := \Pr[Y = i]$.

First, it can be confirmed by basic calculus that
$$\forall p, q \in [0, 1], \eps := |p - q|: \quad
|p\log(1/p) - q \log(1/q)| \leq \eps \log(1/\eps)$$

Thus, defining $\eps_i := |p_i - q_i|$, we have
\begin{align*}
|H(X)-H(Y)| &= |\sum_{i=1}^N p_i \log(1/p_i) - q_i \log(1/q_i)| \\
&\leq \sum_{i=1}^N |p_i \log(1/p_i) - q_i \log(1/q_i)|\\
&\leq \sum_{i=1}^N \eps_i \log(1/\eps_i)
\end{align*}
We have $\sum_i \eps_i = ||X-Y||_1 \leq \eps$, so this quantity is maximized for
$\eps_i = \frac{\eps}{N}$.
Thus,
$$|H(X)-H(Y)| \leq \eps \log(\frac{N}{\eps}) \ .  \qedhere $$
\end{proof}

\subsection{Nice Subset Selection}
Now we can describe how to find ``nice'' sets.
We first approximate the conditional distribution
$\DD_{Z_t | Z_{< t}} \in \Delta(\Delta(\F_q))$ for $Z_1, \dots Z_t \sim \HMM_t$,
by sampling.
This crucially relies on the fact that $\HMM$ is a constructive source (ie,
using the Forward Algorithm).
Then we use Algorithm~\ref{algo:entropy-approx} to estimate conditional
entropies, and select high-entropy indices.

\begin{algorithm}[H]
\caption{{\sc Polar-Preprocess}}
\label{algo:subset}
\begin{algorithmic}[1]
    \Require{$q,k,t \in \N$ with $q$ prime,  $M \in \F_q^{k \x k}$, and Markov source $\HMM$}
    \Comment{$m = k^t, n = m^2$}
    \Ensure{Sets $S_1, S_2, \dots S_m \subseteq [m]$}
\Procedure{Polar-Preprocess}{$q,k,t, M, \HMM$}
	\State $m \gets k^t$; $\gamma \gets \frac{1}{n^{10}}$;
	 $N \gets |T_{\gamma}(\F_q)|$; $R \gets n (N / \gamma)^2$ \Comment{$N \leq poly_q(1/\gamma)$}
    \ForAll{$j \in [m]$}
        \ForAll{$i = 1, 2, \dots, R$}
            \State Sample a sequence $w_i := (y_1, y_2, \dots y_{j-1})$ from $\HMM$.
            \State Compute $\cD_{w_i} \in \Delta(\F_q)$, the distribution of
            $Y_j | Y_{< j} = w_i$,
            using the Forward Algorithm~\ref{algo:forward} for $\HMM$.
        \EndFor
        \State Let $\widetilde{\DD}_{Y | W} \in \Delta(\Delta(F_q))$ be the empirical distribution
        of $\cD_w$, from the samples $\cD_{w_i}$ above. \label{ln:sample}
        \State $\{\hat h_1, \dots \hat h_m\} \gets
        \textsc{ApproxEntropy}(\widetilde{\DD}_{Y|W}, \gamma = \frac{1}{n^4}, t, M)$
        \State $S_j \gets \{i \in [m]: \hat h_i > \frac{1}{n^3}\}$
    \EndFor
    \State \Return $S_1, S_2, \dots S_j$.
\EndProcedure
\end{algorithmic}
\end{algorithm}

\begin{proof}[Proof of Theorem~\ref{thm:preproc}]
First, we claim that for each $j$, the sampling step (Line~\ref{ln:sample}) produces
$\widetilde{\DD}_{Y | W}$
such that
$$\dc(\widetilde{\DD}_{Y | W}, \DD_{Y_j | Y_{< j}}) \leq 2 \gamma$$, except with
probability $\exp(-\Omega(n))$.
Here, $\DD_{Y_j | Y_{< j}}$ denotes the conditional distribution of random
variables $(Y_1, Y_2, \dots Y_j) \sim \HMM_j$.
It is sufficient to show that these distributions are $\gamma$-close,
after rounding both to $\Delta(T_{\gamma})$, a $\gamma$-net of
$\Delta(\Delta(\F_q))$.
Let $R_{\gamma}: \Delta(\F_q) \to T_{\gamma}$ be the function that rounds points
to their nearest net-point. Let this naturally lift to a function
$\bar{R}_\gamma: \Delta(\Delta(\F_q)) \to \Delta(T_\gamma)$.
Now, the net $T_\gamma$ has only $N$ points,
so sampling $n (N / \gamma)^2$ points from $\bar{R}_\gamma(\DD_{Y_j | Y_{< j}})$ will approximate the mass of each point to within $\pm (\gamma / N)$, except with probability
$\exp(-\Omega(n))$.
Thus, we have $\dc(\bar{R}_\gamma(\widetilde{\DD}_{Y | W}),
\bar{R}_\gamma(\DD_{Y_j | Y_{< j}})) \leq \gamma$ with high probability,
and this implies
$\dc(\widetilde{\DD}_{Y | W}, \DD_{Y_j | Y_{< j}}) \leq 2 \gamma$.

Now, consider the following random variables.
Let $(Y, W)$ be
the implicit random variables
associated with $\tilde{\DD}_{Y | W}$,
and let $(Y', W')$ be those associated with
$\DD_{Y_j | Y_{< j}}$.
Let $U := M^{\otimes t} Z$
and $U' := M^{\otimes t} Z'$ where $\{(Z_i, W_i)\}$ are independently
distributed as $(Y, W)$, and $\{(Z'_i, W'_i)\}$ are independently
distributed as $(Y', W')$.
Now, by triangle inequality, we have
$$\dc(\DD_{U_i | U_{< i}}, \DD_{U'_i | U'_{< i}}) \leq 2 \gamma m$$.

Using the above definitions, by Lemma~\ref{lem:entropy-dc},
the conditional entropies $H(U_i | U_{< i})$
and
$H(U'_i | U'_{< i})$
differ by at most $\pm \sqrt{2\gamma m} =
\Oh(1/n^4)$.
Thus, set $S_j$ selected will have all entropies
$$\forall i \in S_j: \bH(\bU_{(i, j)} | \bU_{\prec (i, j)}) = \frac{1}{n^3} \pm
\Oh(\frac{1}{n^4}) = \Theta(\frac{1}{n^3})$$

Finally, by the Polarization Lemma (Lemma~\ref{lem:polarization}),
all but $\eps n$ of the entropies
$\{\bH(\bU_{(i, j)} | \bU_{\prec (i, j)})\}_{i \in S_j}$
are $\geq 1-\eps$.
Thus, the size of $\sum_j |S_j|$ is at most
$$\eps n + \bH(\Ziid) / (1-\eps)
\leq 2 \eps n + \bH(\Ziid)
\leq 2 \eps n + \bH(\Ziid) + \exp(-\Omega(n))
\leq 3\eps n + \bH(\Ziid)
$$
Now, by Claim~\ref{lem:entropy-L1} and the closeness of distributions $\Ziid$ and
$Z$, we have
$\bH(\Ziid) \leq \bH(Z) + n^2 \exp(- \eps m / \tau) < \bH(Z) + \Oh(\eps n)$.
Thus, finally,  $\sum_j |S_j| \leq \bH(Z) + \Oh(\eps n)$ as required.
\end{proof}

\section{Proofs of Theorems~\ref{thm:main}~and~\ref{thm:main-channel}}
\label{sec:proofs}

Combining Theorem~\ref{thm:preproc} (to compute nice sets) with Theorem~\ref{thm:main-analysis} (compressing and decompressing assuming nice sets), Theorem~\ref{thm:main} follows immediately.

\begin{proof}[Proof of Theorem~\ref{thm:main}]
The algorithms claimed are Algorithm~\ref{algo:subset} for preprocessing, Algorithm~\ref{algo:compressor} for compressing and Algorithm~\ref{algo:fast-decompressor} for decompression. Theorem~\ref{thm:preproc} asserts that Algorithm~\ref{algo:subset} returns a nice sequence of sets $S_1,\ldots,S_m$ with all but exponentially small probability in $n$. And Theorem~\ref{thm:main-analysis} asserts that if $S_1,\ldots,S_m$ are nice then Algorithm~\ref{algo:compressor}~and~\ref{algo:fast-decompressor} compress and decompress correctly with high probability over the output of the Markovian source. This yields the theorem.
\end{proof}

Finally we show how Theorem~\ref{thm:main-channel} follows from
Theorem~\ref{thm:main}.

\begin{proof}[Proof of Thereom~\ref{thm:main-channel}]
	Let $H \in \F_q^{s \x n}$ be the matrix specifying the (linear) compression
	scheme given by the Preprocessing Algorithm in Theorem~\ref{thm:main},
	when applied to Markov source $\HMM$.
	The code $C$ for the additive Markov Channel $\Ch_{\HMM}$ is simply specified by the
	nullspace of $H$, ie encoding is given by $C(x) := Nx$ where $N \in \F_q^{n \x
		n-s}$ spans $Null(A)$.

	Note that due to the structure of $H$, a nullspace matrix $N$ can be applied
	in $\Oh_q(n \log n)$ time.
	In particular, $H$ is a subset of rows of the block-diagonal matrix
	$P \in \F_q^{n \x n}$, where each $\sqrt{n} \x \sqrt{n}$ block is
	the tensor-power $M^{\otimes t}$.
	Thus, $P^{-1}$ is also block-diagonal with blocks $(M^{-1})^{\otimes t}$, and so
	can be applied in time $\Oh_q(n \log n)$.
	The matrix $N$ can be chosen as just a subset of columns of $P^{-1}$, and hence
	can also be applied in time $\Oh_q(n \log n)$.

	Let $y_1, y_2, \dots y_n \in \F_q$ be distributed according to $\HMM$, and $y :=
	(y_1, \dots y_n) \in \F_q^n$.
	To decode from $z = Nx + y$, the decoder first applies $H$ (by running the
	compression algorithm of Theorem~\ref{thm:main}), to compute $Hz = HNx + Hy = Hy$.
	Then, the decoder runs the decompression algorithm of Theorem~\ref{thm:main} on
	$Hy$ to determine $y$.
	Finally, the decoder can compute $y - z$ to find the codeword sent $(Nx)$,
	and thus determine $x$.
	(Again using the structure of $P$, as above, to determine $x$ from
	$Nx$ in $\Oh_q(n \log n)$ time).
\end{proof}

\bibliographystyle{plain}
\bibliography{polar-refs}

\begin{appendix}
\section{Forward Algorithm}
\label{app:forward-algo}

\begin{algorithm}[H]
\caption{Forward Algorithm}
\label{algo:forward}
\begin{algorithmic}[1]
    \Require{$n \in \N$. Markov source $\HMM$ with state-space $[\ell]$, alphabet $\Sigma$,
    stationary distribution $\pi \in \Delta([\ell])$,
    transition matrix $\Pi \in \R^{\ell \x \ell}$},
    and output distributions $\{\cS_i \in \Delta(\Sigma)\}_{i \in [\ell]}$.
    And $y = (y_1, y_2, \dots y_{n-1})$ for $y_i \in \Sigma$.
    \Ensure{Distribution $Y_n \in \Delta(\Sigma)$}
\Procedure{ForwardInfer}{$\HMM=(\ell, \Sigma, \pi, \Pi, \{\cS_i\}), n, y$}
    \State $s_0 \gets \pi$.
    \ForAll{$t = 1, 2, \dots {n-1}$}
        \State Define $s_t \in \Delta([\ell])$ by
        $s_t(i) \gets
        \frac{
        (\Pi s_{t-1})_i \cdot \cS_i(y_t)
        }
        {
        \sum_{j \in [\ell]}
        (\Pi s_{t-1})_j \cdot \cS_j(y_t)
        }$
        \Comment{Treating $s_{t-1}$ as a vector in the probability simplex
        embedded in $\R^\ell$}
    \EndFor
    \State $s_n \gets \Pi s_{n-1}$.
    \State \Return The distribution $Y_n := \E_{i \sim s_n}[\cS_i]$.
\EndProcedure
\end{algorithmic}
\end{algorithm}

\begin{claim}
For every Markov source $\HMM = (\ell, \Sigma, \pi, \Pi, \{\cS_i\})$,
let random variables $Y_1, \dots Y_n \sim \HMM_n$.
For every setting $y = (y_1, y_2, \dots y_{n-1})$ for $y_i \in \Sigma$,
let $\cD_{Y_n | Y_{< n} = y}$ denote the distribution of $Y_n$ conditioned on
$Y_{< n} = y$.
Then,
$$
\textsc{ForwardInfer}(\HMM, n, y)
\equiv
\cD_{Y_n | Y_{< n} = y}
$$
\end{claim}

This follows inductively, from the fact that
$s_t$ as maintained by the algorithm is exactly
the distribution of $S_t | \{Y_{\leq t} = y_{\leq t}\}$,
where $S_t$ is the hidden state of $\HMM$ after $t$ steps.

\section{Connection to Learning Parity with Noise}
\label{app:LPN}

The problem of learning parity with noise (LPN) is the following. Fix an
(unknown) string $a \in \F_2^\ell$ and $\eta > 0$ and let $D_{a,\eta}$ be the
distribution on $\F_2^{\ell+1}$ whose samples $(x,y)$ are generated as follows:
Draw $x \in \F_2^\ell$ uniformly and let $z\in Bern(\eta)$ be drawn independent of $x$ and let $y = \langle a,x \rangle + z$ where  $ \langle a,x \rangle = \sum_{i=1}^\ell a_i x_i $. Given samples $(x_1,y_1),\ldots,(x_m,y_m)$ drawn i.i.d. from such a distribution, the LPN problem is the task of ``learning'' $a$.

It is well known that $a$ is uniquely determined by $O(\ell)$ samples (i.e., $m = O(\ell)$) where the constant in the $O(\cdot)$ depends on $\eta < 1/2$. However no polynomial time algorithms are known that work with $m = \poly(\ell)$ and determine $a$ for any $\eta > 0$ and indeed this is believed to be a hard task in learning. We refer to this hardness assumption as the LPN hypothesis.

The connection to learning Markovian sources comes from the fact that samples from the distribution $D_{a,\eta}$ can be generated by an $O(\ell)$-state Markov chain. (Briefly the states are indexed $(i,b,c)$ indicating $\sum_{j=1}^{i-1} a_j x_j = b$ and $x_i = c$. For $i < \ell$ the state $(i,b,c)$ outputs $c$ and transtions to $(i+1,b+c,0)$ w.p. 1/2 and to $(i+1,b+c,1)$ w.p. 1/2. When $i = \ell$, the state $(i,b,c)$ outputs $(c,b+c)$ w.p. $1-\eta$ and $(c,b+c+1)$ w.p. $\eta$ and transitions to $(1,0,0)$ w.p. 1/2 and to $(1,0,1)$ w.p. $1/2$.)
The entropy of this source is $(\ell+H(\eta))/(\ell+1)$. A compression with $\eps = (1 - H(\eta))/(2(\ell+1))$ with $\poly(\ell/\eps)$ samples from the source would distinguish this source from purely random strings which in turn enables recovery of $a$, contradicting the LPN hypothesis.

We thus conclude that compressing an {\em unknown} Markov source with number of samples that is a polynomial in the mixing time and the inverse of the gap to capacity contradicts the LPN hypothesis.

\end{appendix}

\end{document}